\documentclass[twocolumn,aps,prl,superscriptaddress]{revtex4-2}
\usepackage{chngcntr}
\usepackage{times}
\usepackage{amssymb}
\usepackage{amsthm}
\usepackage{amsmath}
\usepackage{mathrsfs}
\usepackage{graphicx}
\usepackage{tikz}
\usepackage{float}
\usepackage[caption=false]{subfig}
\usepackage[pdftex,colorlinks,linkcolor={red!75!black},citecolor={red!75!black},urlcolor={red!75!black}]{hyperref}
\usepackage{mathtools}
\usepackage{verbatim}
\usepackage{epstopdf}
\usepackage{xcolor}
\usepackage{physics}
\usepackage{bm}
\usepackage{soul}
\usepackage{ulem}

\renewcommand{\Im}{\operatorname{Im}}
\newcommand{\clr}{\color{red!75!black}}

\newcommand{\tx}[1]{\text{#1}}
\newcommand{\bs}[1]{\boldsymbol{#1}}

\newtheorem{theorem}{Theorem}
\newtheorem{lemma}{Lemma}

\newtheorem{corollary}{Corollary}

\begin{document}

\title{Universal spectral moment theorem and its applications in non-Hermitian systems}

\author{Nan Cheng}
\author{Chang Shu}
\author{Kai Zhang}
\author{Xiaoming Mao}
\author{Kai Sun}
\affiliation{Department of Physics, University of Michigan, Ann Arbor, MI 48109, USA}

\begin{abstract}
The high sensitivity of the spectrum and wavefunctions to boundary conditions, termed the non-Hermitian skin effect, represents a fundamental aspect of non-Hermitian systems. While it endows non-Hermitian systems with unprecedented physical properties, it presents notable obstacles in grasping universal properties that are robust against microscopic details and boundary conditions. 
In this Letter, we introduce a pivotal theorem: in the thermodynamic limit, for any non-Hermitian systems with finite-range interactions, all spectral moments are invariant quantities, independent of boundary conditions, posing strong constraints on the spectrum. 
Utilizing this invariance, we propose a new criterion for bulk dynamical phases based on experimentally observable features and applicable to any dimensions and any boundary conditions.
Based on this criterion, we define the bulk dispersive-to-proliferative phase transition, which is distinct from the real-to-complex spectral transition and contrary to traditional expectations.
We verify these findings in 1D and 2D lattice models.
\end{abstract}

\maketitle

\noindent{\it \clr Introduction.}---When a system interacts with external environments, the use of a non-Hermitian Hamiltonian becomes an efficient description and leads to a realm of new discoveries~\cite{Rotter2009,Moiseyev2011,Ashida2020}. 
The non-Hermitian elements manifest differently in various physical setups, for example, imbalanced mode damping in optical and acoustic systems~\cite{FengLiang2017,Longhi2017Review,Ganainy2018,Huang2023NRP}, odd viscosity and elasticity in mechanical systems~\cite{Vincenzo2020NP,Vincenzo2020PRL,ZhouDi2020PRR,Banerjee2021PRL,ShankarNature}, quasi-particle excitations with finite lifetime in condensed matter~\cite{ShenHTPRL2018,FLPRL2020}, time evolution of observables in open-quantum systems~\cite{SongFei2019,ChuCH2020PRR,Ueda2021PRL}, and dynamics of species population in biological systems~\cite{Nelson1998PRE,Nelson2016PRE}. 
The non-Hermitian Hamiltonian enables complex eigenvalues, giving rise to a myriad of intriguing phenomena not found in conservative systems~\cite{McDonaldPRX2018,Ashvin2019PRL,LQPRL2022,Kai2023PRL,Diego2019PRL,YYFPRL2020,XueWT2021PRB,LiLH2021NC,Longhi2022_PRL}. 

One central topic in non-Hermitian band systems~\cite{FuLiang2018_PRL, GongPRX2018, SatoPRX2019, Bergholtz2021_RMP} is the non-Hermitian skin effect (NHSE)~\cite{Yao2018, Murakami2019_PRL, ChingHua2019,Slager2020PRL,Kai2020, Okuma2020_PRL, ZSaGBZPRL, LeeCHReview}, where a large number of bulk wavefunctions localize at open boundaries. 
A key feature of NHSE is its high spectral sensitivity to boundary conditions~\cite{LeeModel2016, Torres2018, Xiong2018}. 
It is generally observed that the spectrum is dramatically reshaped as the boundary conditions change from periodic to open.
In two and higher dimensions, the spectrum exhibits even more complex characteristics~\cite{WangZhong2018,LeeCH2019_PRL,Kai2020,Okuma2020_PRL,Nori2019,Song_2022,HuiJiang2022,Murakami2022,wang2022amoeba,hu2023nonhermitian}. 
The spectral density distribution also depends on different open boundary conditions (OBC) geometries~\cite{Kai2022NC,Wang2022NC,QYZhouN2023C,KunD2023PRL,WanSciB,EdgeTheory2023}. 

Despite the exotic physical properties conferred by spectral sensitivity in non-Hermitian systems, the full understanding of this phenomenon remains elusive. This sensitivity to boundary conditions cannot be an arbitrary rearrangement of energy spectra; they must adhere to fundamental principles that are impervious to boundary conditions. The rationale for expecting such universality rests on the premise that, in systems with local (finite-range) interactions, altering boundary conditions only modifies a sub-extensive part of the system, whose volume compared to the bulk approaches zero in the thermodynamic limit. Consequently, there must exist pivotal characteristics dictated solely by the bulk, immune to any variations in boundary conditions. Some nascent insights into such bulk-dictated properties have surfaced recently; for instance, in certain systems of NHSE, although the (right) eigenstates display high boundary sensitivity, their local density of states are uniform in the bulk and insensitive to boundary conditions~\cite{YYFPRL2020,Okuma2021PRL}. Additionally, short-term wavepacket evolution in the bulk appears boundary-agnostic~\cite{MaoLiang2021}. Notwithstanding these exciting findings, underlying invariants and universal principles are yet to be unraveled.

In this Letter, we introduce and prove a universal spectral moment theorem, applicable to any systems with finite-range couplings---Hermitian or non-Hermitian. We demonstrate that in the thermodynamic limit, despite potentially dramatic shifts in their energy spectrum, all moments of the spectrum are determined entirely by the bulk and are invariant with respect to boundary conditions. For Hermitian systems, this theorem validates the longstanding thermodynamic principle that the density of states, when normalized to the bulk volume, is insensitive to boundary conditions. In the context of non-Hermitian systems, although energy spectra and densities of states may change significantly upon altering boundary conditions, this sensitivity is ultimately constrained by the bulk.

To demonstrate one application of the theorem, we study the time evolution of wave packets in $\mathcal{PT}$-symmetric non-Hermitian systems. Previous research has thoroughly investigated the phase transition from the $\mathcal{PT}$-exact phase, featuring entirely real eigenvalues, to the $\mathcal{PT}$-broken phase, in which some eigenvalues turn complex. However, for systems in the thermodynamic limit, our understanding on the precise effects of non-Hermiticity on $\mathcal{PT}$-symmetric systems remains limited, due to the extreme sensitivity of the eigenspectrum to the boundary conditions. In this context, we introduce the concept of the dispersive-to-proliferative transition, which encapsulates the competition between dispersiveness and non-unitary time evolution. By applying the theorem above, we show that this transition sets a universal upper bound on the real-to-complex transition of the eigenvalue spectrum, completely unaffected by the choice of boundary conditions or any sub-extensive perturbations. 

\noindent{\it {\clr The universal spectral moment theorem.---}}
Here we present the universal spectral moment theorem using lattice Hamiltonians, while these conclusions can also be generalized to continuous models by appropriately taking the continuum limit. 
We first define some notational conventions used in this paper. 
Let $\Omega$ be a bounded connected open region in $\mathbb{R}^d$, $\Gamma$ be a fixed infinite lattice in $\mathbb{R}^d$, $V$ be the volume of the Brillouin zone (BZ), $r\Gamma$ be the lattice obtained by scaling the lattice $\Gamma$ by a factor of $r$, $H$ be a real space periodic non-Hermitian lattice Hamiltonian with finite interaction range defined on the infinite lattice $\Gamma$ with Bloch Hamiltonian $H(\bs{k})$, $\bs{k}\in \mathrm{BZ}$. 
Without loss of generality, we assume that each node in the unit cell has only one degree of freedom and the number of nodes in a unit cell is $m$.  
Let $H_{r}$ be a lattice Hamiltonian with $N_{r}$ degrees of freedom defined on a finite lattice $\Omega\cap r\Gamma$ with the same interaction parameter (same nearest neighbor hopping, etc.) as $H$. 
As we decrease $r$ toward $0$, $H_{r}$ remains defined in the same open region $\Omega$, yet the lattice mesh becomes increasingly dense [Figs.~\ref{fig:1}(a) and (b)], with the continuum limit corresponding to $r \to 0$. For a lattice model, the limit $r \to 0$ is essentially equivalent to maintaining a constant lattice spacing while scaling the size of the open region to infinity, i.e., the thermodynamic limit. 

Let $\rho_{\Omega}(E)$ be the normalized spectral density of the open-boundary Hamiltonian in the continuum limit, defined as $\rho_{\Omega}(E)=\lim_{l\rightarrow 0}\lim_{r\rightarrow 0}N(E,l,r)/l^{2}N_{r}$, where $N(E,l,r)$ counts the states within a square of area $l^2$ centered at $E$ in the complex energy plane. The integral of $\rho_{\Omega}(E)$ over the entire energy plane is $1$. 
Although the spectral density $\rho_{\Omega}(E)$ itself may depend on the boundary geometry, the spectral moments are invariant, as stated in the following universal spectral moment theorem. 
\begin{theorem}\label{thm:spectralmoment}
    For any positive integer $n$, the $n^{th}$ moment of the normalized density of states $\rho_{\Omega}(z)$ in the continuum limit is independent of the boundary condition and is related to the Bloch Hamiltonian $H(\bs{k})$ by the following formula
    \begin{equation}\label{eq:spectralmoment}
        \int_{E\in\mathbb{C}}E^{n}\rho_{\Omega}(E)\,\dd S=\frac{1}{mV}\int_{\bs{k}\in \text{BZ}}\Tr(H(\bs{k})^{n})\,\dd\bs{k},
    \end{equation}
    where $\dd S$ is the area element in the complex-energy plane.
\end{theorem}
Theorem~\ref{thm:spectralmoment} states that the arbitrary order-$n$ spectral moment of the OBC normalized density of states is an intrinsic property independent of boundary conditions of the non-Hermitian Hamiltonian. 
It's worth noting that the spectral moments, although defined by the OBC Hamiltonian, can be determined by only solving the Bloch Hamiltonian $H(\bs{k})$, thereby making them easily computable. 
For lattice Hamiltonians with finite-range couplings, the complex-valued OBC spectrum covers a bounded region, and $\rho_{\Omega}(E)$ is zero for all sufficiently large $|E|$. 
One straightforward application is that when the OBC spectrum is real, $\rho_{\Omega}(E)$ can be completely determined by its spectral moments (related to Hausdorff moment problem~\cite{shohat1943problem}). 
When the spectrum becomes complex, these spectral moments in general do not uniquely determine $\rho_{\Omega}(E)$. We need the information of all mixed moments $\int_{E\in\mathbb{C}}E^{a}(E^{*})^{b}\rho_{\Omega}(E)\,\dd S$ to fully determine $\rho_{\Omega}(E)$.  
Now we provide a proof of Theorem~\ref{thm:spectralmoment}. 
\begin{figure}[t]
    \centering \includegraphics[width=1\linewidth]{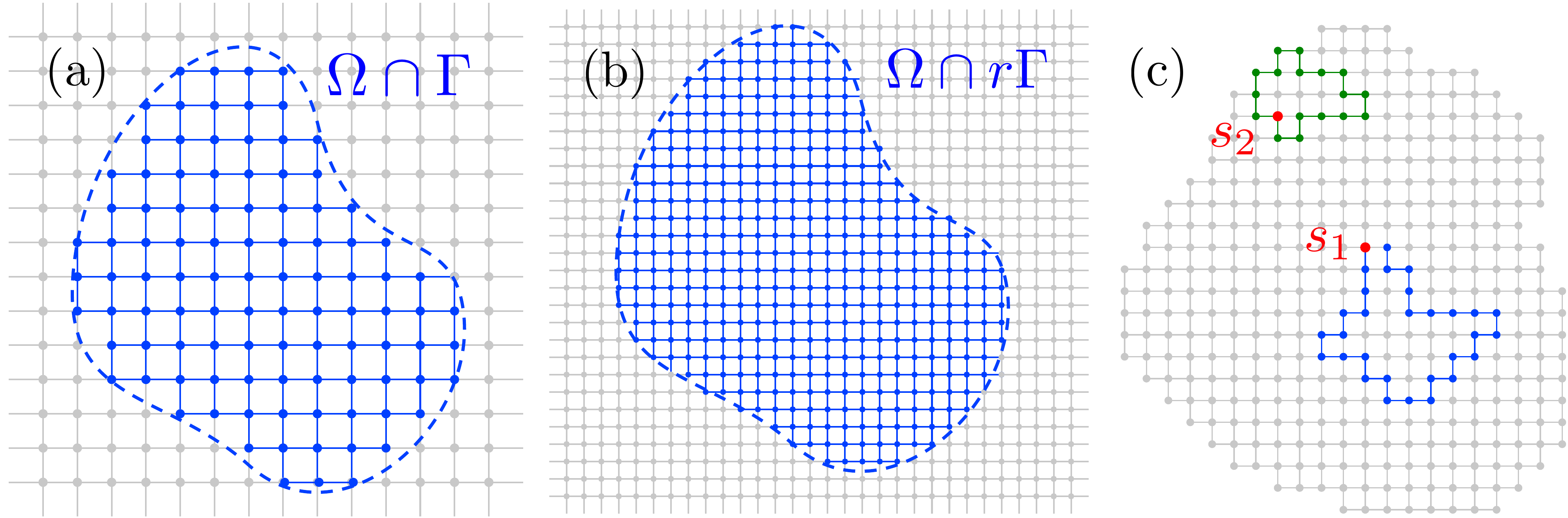}
    \caption{Lattice scaling and loops in $\Omega\cap r\Gamma$. 
    (a) Region $\Omega$ embedded in the background lattice $\Gamma$.
    (b) Same region $\Omega$ with scaled background lattice $r\Gamma$ ($r<1$).
    (c) Loops starting from bulk point $s_1$ and from point $s_2$ near the boundary.}
    \label{fig:1}
\end{figure}
\begin{proof} 
    By the definition of the density of states $\rho_{\Omega}(E)$, we have
    \begin{equation}
        \lim_{r\rightarrow 0}\frac{\Tr H_{r}^{n}}{N_{r}}=\int_{E\in\mathbb{C}}E^{n}\rho_{\Omega}(E)dS.
    \end{equation}     
    Given a node $s$ in a finite size lattice ($r\Gamma\cap\Omega$), we have
    \begin{equation}\label{eq:loopcounting}
        (H_{r}^{n})_{ss}=\sum_{i_{1},...,i_{n-1}}H_{si_{n-1}}...H_{i_{2}i_{1}}H_{i_{1}s}=\sum_{L_{s}}\omega(L_{s}),
    \end{equation}
    where $L_{s}$ indicates a loop starting and ending at node $s$ [Fig.~\ref{fig:1}], and the last summation is over all loops with the weight $\omega(L_{s})$ being the product of the hopping strength on the loop $L_{s}$. 
    From Eq.~\eqref{eq:loopcounting} we see that if a node $s$ is deep in the bulk [Fig.~\ref{fig:1}(c)], all loops $L_{s}$ cannot touch the lattice boundary, and thus $(H_{r}^{n})_{ss}$ largely depends on the bulk. 
    Let $H_{R}$ be the real space Hamiltonian with the same hopping parameters on a large torus of corresponding dimension such that the number of nodes $R$ along one direction is much larger than $n$. 
    Let $n_{R}$ be the number of unit cells contained in the graph defining $H_{R}$.
    Using the fact that for any given $n$, the portion of nodes no farther than $n/2$ hopping steps away from the boundary in the set $r\Gamma\cap\Omega$ tends to zero [see Fig.~\ref{fig:1}(c) and Appendix A] and $(H_{r}^{n})_{ss}$ doesn't depend on the boundary when node $s$ is in the bulk, we have
    \begin{equation}\label{eq:obctotorus}
        \lim_{r\rightarrow 0}\frac{\Tr H_{r}^{n}}{N_{r}}=\frac{\Tr H_{R}^{n}}{mn_{R}}.
    \end{equation}
    Since the left-hand side doesn't depend on $R$, we can apply Bloch's theorem to block diagonalize $H_{R}$ and take $R\rightarrow\infty$ limit on both sides of Eq.~\eqref{eq:obctotorus}.
    It follows that
    \begin{equation}\label{eq:limittraceidentity}
        \lim_{r\rightarrow 0}\frac{\Tr H_{r}^{n}}{N_{r}}=\lim_{R\rightarrow\infty}\frac{\Tr H_{R}^{n}}{mn_{R}}=\frac{1}{mV}\int_{\bs{k}\in BZ}\Tr(H(\bs{k})^{n})\,\dd\bs{k},
    \end{equation}
    which completes the proof of Theorem~\ref{thm:spectralmoment}. 
\end{proof}
From this proof, we can also see that the spectral moments of the non-Hermitian Hamiltonian under periodic boundary conditions (PBC) are the same as those under OBC in the continuum limit. 
Consequently, the spectral moments are indeed independent of boundary conditions. 
It is worth noting that if $F(z)=\sum_{n=0}^{\infty}a_{n}z^{n}$ is an analytical function with no poles in the PBC and OBC spectrum, we have [see Appendix B for a proof of the case where $F(z)=e^{-i zt}$] 
\begin{equation}\label{eq:szegothm}
    \lim_{r\rightarrow 0}\frac{\Tr\,F(H_{r})}{N_{r}}=\frac{1}{mV}\int_{\bs{k}\in BZ}\Tr\,F[H(\bs{k})]\,\dd\bs{k}.
\end{equation}
This result is fully consistent with Szegö's limit theorem~\cite{WIDOM1980182}, and this conclusion can be generalized to encompass arbitrary boundary conditions as well as any form of sub-extensive perturbations or disorders that may occur in the system.

To demonstrate Theorem~\ref{thm:spectralmoment}, we calculate the spectrum of a non-Hermitian lattice model in two dimensions and show that the normalized density of states depends on boundary conditions while the spectral moments don't. 
Consider a non-Hermitian tight binding model as illustrated in Fig.~\ref{fig:2}(a). 
We calculate the periodic-boundary spectrum, i.e., $H(\bs{k})$ for all $\bs{k}$ in the Brillouin zone, and the open-boundary eigenvalues with system size of $L_x=L_y=100$ [Fig.~\ref{fig:2}(b)]. 
The spectral density is drastically different under OBC and PBC [Fig.~\ref{fig:2}(c)]. 
However, their spectral moments coincide when the system size tends to continuum limit ($r\to 0$), as shown in Fig.~\ref{fig:2}(d). 
The spectral moments with distinct boundary conditions converge at the rate of $r$, and this holds true regardless of system dimensions [see Appendix A]. 

\begin{figure}[t]
    \centering
    \includegraphics[width=1\linewidth]{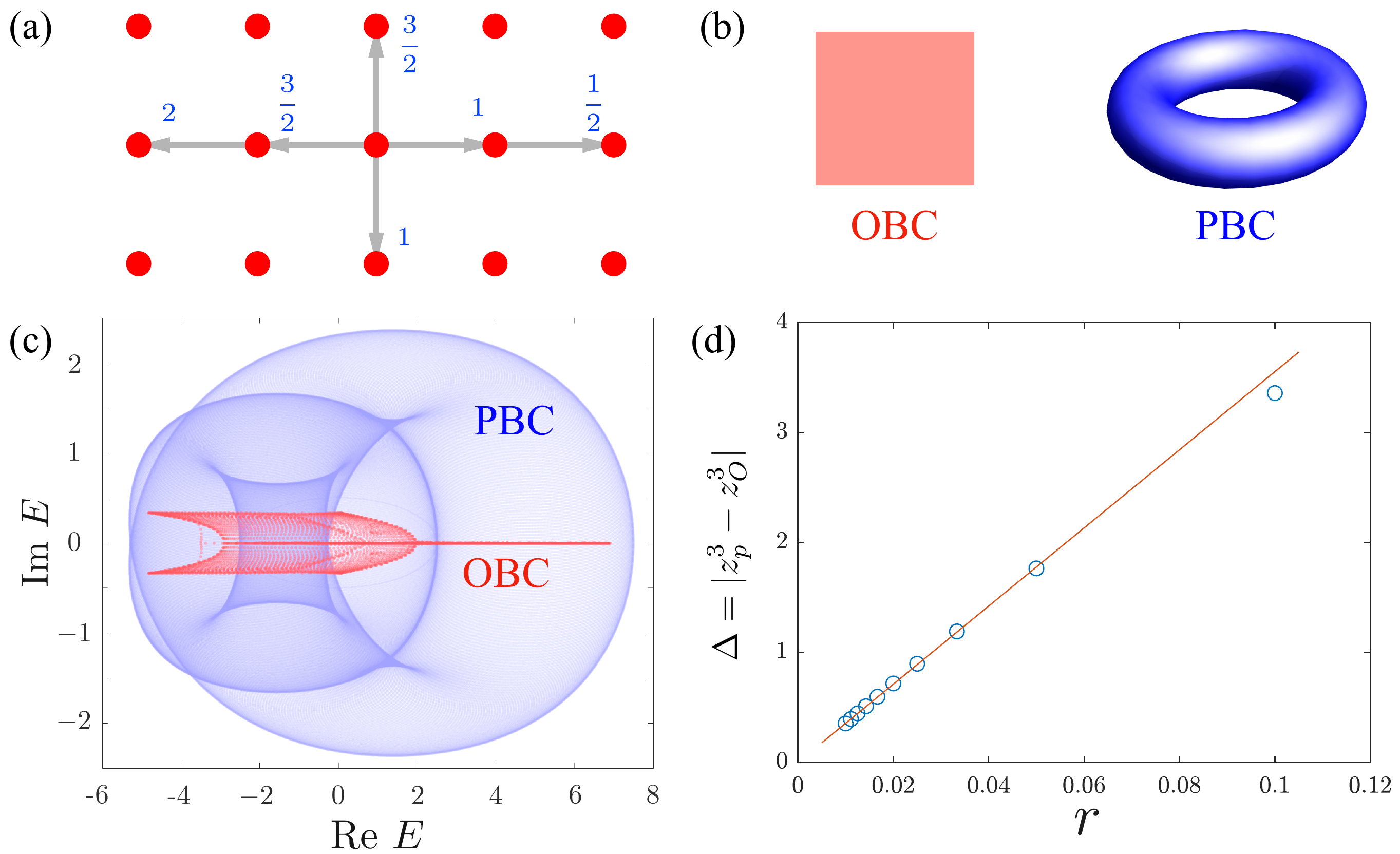}\caption{Invariance of spectral moments in 2D. 
    (a) The 2D non-Hermitian lattice model; 
    (b) OBC and PBC geometry;  (c) The OBC spectrum (red) and PBC spectrum (Blue) are drastically different; (d) Convergence of $3^{rd}$ spectral moment as a function of $r$. $z_{P}^3$ is the $3^{rd}$ spectral moment of an infinite system (right-hand side of Eq.~\eqref{eq:spectralmoment}), $z_{O}^3$ is the $3^{rd}$ spectral moment of large systems under OBC.} 
    \label{fig:2}
\end{figure}

\noindent{\it {\clr Dispersive-to-proliferative transition in $\mathcal{PT}$-symmetric non-Hermitian systems.---}}
The theorem discussed above not only imposes fundamental constraints on the energy spectrum, but also has practical relevance for experimental studies. To illustrate this, we investigate a $d$-dimensional $\mathcal{PT}$-symmetric system under arbitrary boundary conditions and examine its bulk wave dynamics, which is a standard experimental technique for probing non-Hermitian systems~\cite{XuePeng2021PRL}.

In $\mathcal{PT}$-symmetric systems, as the strength of non-Hermiticity $\gamma$ increases beyond a certain threshold $\gamma^*$, a phase transition will often occur, where the energy spectrum of the systems transitions from real to complex, known as the $\mathcal{PT}$-exact to $\mathcal{PT}$-broken transition. 
Depending on the system, the value of $\gamma^*$ may be zero or finite. More importantly, this threshold $\gamma^*$ can also vary with boundary conditions, as the energy eigenvalues are highly sensitive to them. However, in realistic experimental systems with large sizes, modifying boundary conditions only affects a negligible fraction of the system. Thus, except in the vicinity of the boundaries, the time dynamics should remain uninfluenced by boundary configurations. 

With this understanding we define this transition as follows. Prepare a $\delta$-function wave-packet at one bulk point (e.g., the origin), and measure its amplitude at the same point at a later time $t$. This measurement probes the amplitude of the equal-position correlation function, which takes the following form in the thermodynamic limit and can be viewed as the time-evolution operator averaged over all eigenvalues, 
\begin{equation}\label{eq:avgpropeigval}
    \bar{G}(t) \equiv \lim_{r\rightarrow 0}\bar{G}_r(t) = \lim_{r\rightarrow 0} \frac{1}{N_r} \Tr\;e^{-iH_{r}t}.
\end{equation}
Due to the $\mathcal{PT}$ symmetry, $|\bar{G}(t)|=|\bar{G}(-t)|$. Here we focus on the fate of $|\bar{G}(t)|$ at large time, i.e., whether it decays to zero or diverges towards infinity. 

In the thermodynamic limit, the time-evolution of $|\bar{G}(t)|$ is mainly governed by two competing effects: dispersiveness and non-unitary time evolution. At the Hermitian limit, $|\bar{G}(t)|$ decays as a power-law function of $t$ due to dipersiveness. In contrast, at strong non-Hermiticity, non-unitary time evolution makes $|\bar{G}(t)|$ to grow exponentially at large $t$. Thus, as we gradually increase the strength of non-Hermiticity ($\gamma$), a phase transition between the dispersive phase, $|\bar{G}(t)|\to 0$, and the proliferative phase, $|\bar{G}(t)|\to \infty$, shall arise, and the critical strengths of non-Hermiticity will be labeled $\gamma_c$. As will be shown below, in contrast to the real-to-complex spectrum transition, where $\gamma^*$ is often zero or sensitive to boundary conditions, the critical value $\gamma_c$ for this dispersive-to-proliferative transition is typically finite, and it is completely independent of boundary conditions at the thermodynamic limit.

Notably, while the two phase transitions (spectrum versus dispersiveness) are distinct, we will establish a deep connection between them by proving that $\gamma_c$ sets a universal upper bound for $\gamma^*$, as illustrated in Fig.~\ref{fig:3}(a). That means, although the value of $\gamma^*$ may fluctuate with different boundary conditions, it will never surpass $\gamma_c$. It's important to recognize that our analysis here primarily focuses on the regime of small $\gamma$, omitting the re-entry effect --- a phenomenon where the $\mathcal{PT}$-exact phase (or dispersive phase) can reappear at high values of $\gamma$. In more general scenarios, our findings indicate that the proliferative phase is necessarily contained within the complex spectrum phase, regardless of boundary conditions.

\noindent{\it {\clr Relations and differences between $\gamma^*$ and $\gamma_c$.---}}
From Eq.~\eqref{eq:szegothm} and Eq.~\eqref{eq:avgpropeigval}, we obtain 
\begin{equation}\label{eq:keyequality}
    \bar{G}(t) =\frac{1}{mV}\sum_{j=1}^{m}\int_{\bs{k} \in \mathrm{BZ}}e^{-i\lambda_{j}(\bs{k})t}\,\dd\bs{k},
\end{equation}
where $\lambda_{j}(\bs{k})$ denotes the $j${-}th band of the Bloch Hamiltonian $H(\bs{k})$.
To demonstrate the connection between $\bar{G}(t)$ and complex spectrum, we prove the following conclusion [A rigorous proof of the following Corollary is in Appendix C.]
\begin{corollary}\label{cor:complexityofspectrum}
    If we have
    \begin{equation}\label{eq:definitionofU}
        \abs{\bar{G}(t)} = \frac{1}{mV}\left|\sum_{j=1}^{m}\int_{\bs{k}\in \tx{BZ}}e^{-i\lambda_{j}(\bs{k})t}\,\dd\bs{k}\right|>1,
    \end{equation}
    for some $t\in\mathbb{R}$,
    $H_{r}$ has complex eigenvalues in the continuum limit.
    That is, there exist $r_{0}>0$ and $\epsilon>0$ such that $H_{r}$ has at least one complex eigenvalue whose absolute value of imaginary part is greater than $\epsilon$ for all $0<r<r_{0}$.
\end{corollary}
This Corollary implies that real OBC spectrum is a sufficient, rather than necessary, condition for dispersiveness, and thus a generic relation between real OBC spectrum and dispersiveness is illustrated in Fig.~\ref{fig:3}(a), where $\gamma$ denotes a general non-Hermitian parameter. 
We emphasize that the dispersive-to-prolifrative transition $\gamma_{c}$ can be detected through short-time bulk dynamics in experiments, and is boundary-agnostic in the thermodynamic limit, making it an intrinsic phase transition. In contrast, the OBC spectrum transition point $\gamma^{*}$ is sensitive to the boundary, as it results from the interplay between the dynamics and the system's boundary. 

\begin{figure}[t]
    \centering
    \includegraphics[width=1\linewidth]{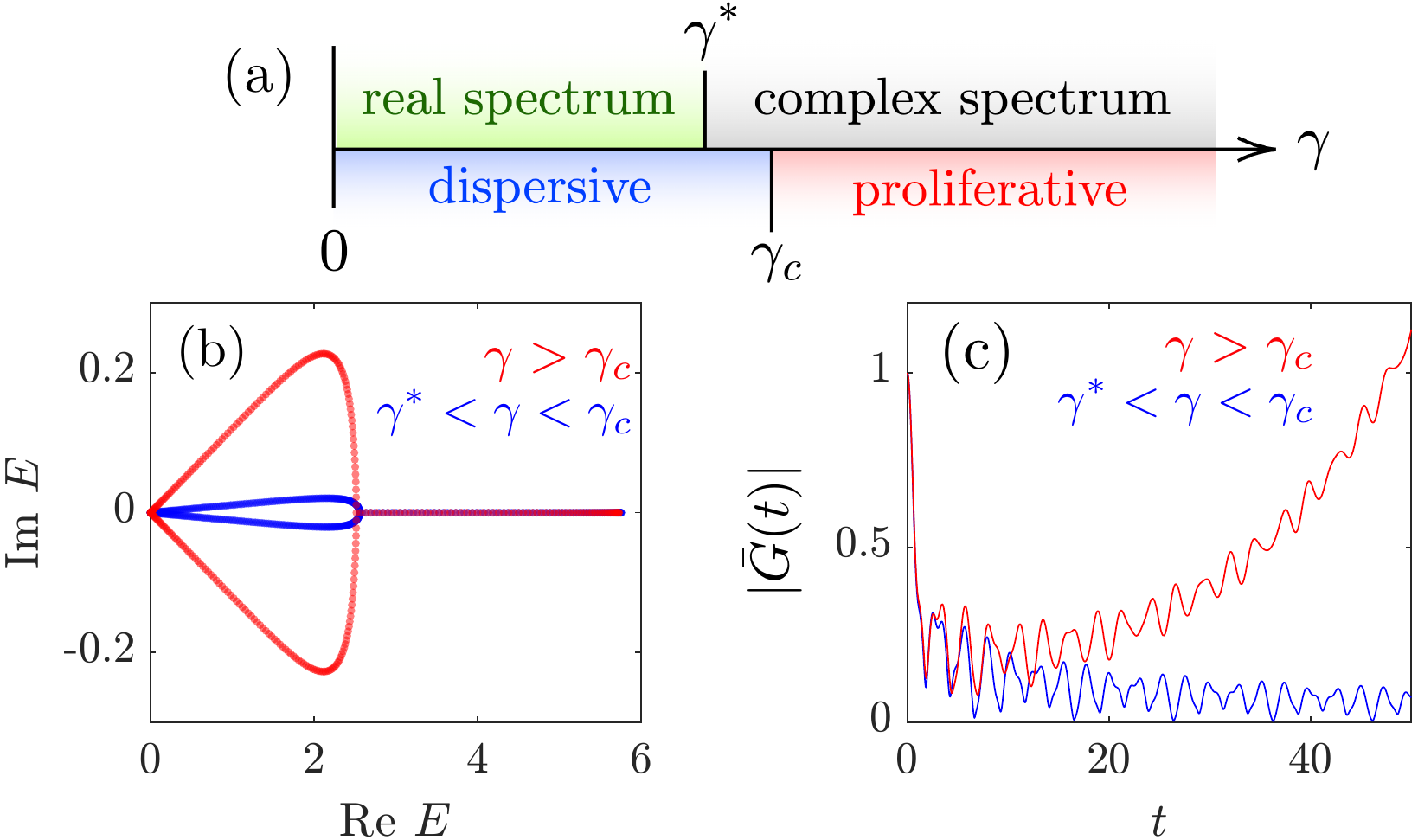}\caption{OBC spectrum and dispersiveness. (a) General relation between real OBC spectrum and dispersiveness, where $\gamma$ denotes general non-Hermitian parameters and $\gamma^{*}\leq\gamma_{c}$; 
    (b) The spectrum for $\gamma^{*}<\gamma=0.01<\gamma_c$ and $\gamma=0.1>\gamma_c$ are all complex; 
    (c) $\bar{G}(t)$ for $\gamma=0.01$ is bounded but grows exponentially at large $t$ for $\gamma=0.1$.}
    \label{fig:3}
\end{figure}

Here, we demonstrate the relation between $\gamma_c$ and $\gamma^*$ by considering the following one-dimensional non-Bloch $\mathcal{PT}$-symmetric Hamiltonian:
\begin{equation}\label{eq:PTexample}
    \mathcal{H}(z)= ((1-\gamma)z + (1+\gamma)z^{-1} + \alpha(z^2+z^{-2}))^2,
\end{equation}
where $z:=e^{ik}$ and $\gamma\geq0$ is the non-Hermitian parameter and $\alpha\geq 0$. 
For $\alpha=0.2$, the OBC spectrum splits from real to complex upon turning on the non-Hermiticity ($\gamma>0$), indicating the transition point $\gamma^{*}=0$ for this Hamiltonian. 
However, we will show that the dispersive-to-proliferative transition at $\gamma_{c}$ does not coincide with $\gamma^{\ast}$; specifically, $\gamma_{c} \approx 0.053 > \gamma^{*}$ for this model. 
The mismatch between $\gamma_c$ and $\gamma^{\ast}$ will lead to a key physical consequence that extends beyond conventional expectation: even though the OBC spectrum is complex-valued at some $\gamma$ (where $\gamma^{\ast} < \gamma < \gamma_c$), the bulk dynamics characterized by Eq.~(\ref{eq:avgpropeigval}) converge in the thermodynamic limit, as illustrated in Fig.~\ref{fig:3}(c). 

\begin{figure}[b]
    \centering
    \includegraphics[width=1\linewidth]{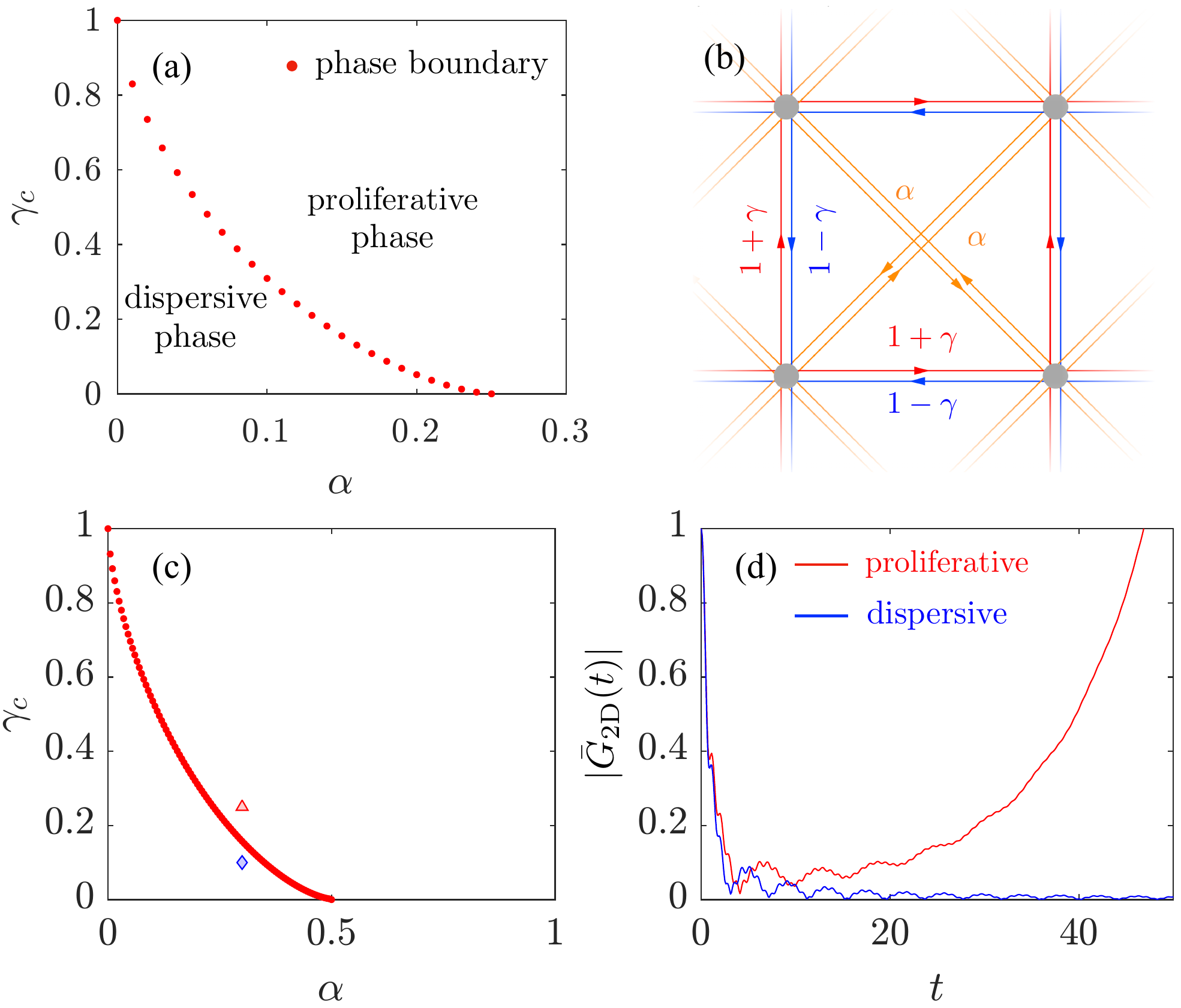}\caption{dispersive-to-proliferative transition phase diagrams in 1D and 2D. (a) dispersive-to-proliferative phase diagram for the 1D Hamiltonian $\mathcal{H}_{\text{1D}}(z)$. Red dots are the computed phase boundary using the contour integral method.
    (b) The illustration for the 2D lattice Hamiltonian given by  Eq.~\eqref{eq:PT2Dexample}. 
    (c) dispersive-to-proliferative phase diagram of the 2D Hamiltonian. The non-Bloch $\mathcal{PT}$ transition point for $\alpha<0.5$ is a finite value. 
    (d) When $\alpha=0.3$, $\bar{G}_{\text{2D}}(t)$ for  $\gamma=0.1<\gamma_c$ [signaled by a diamond in (c)] is bounded; while it grows exponentially at large $t$ for $\gamma=0.25>\gamma_c$ [indicated by a triangle in (c)].} 
    \label{fig:4}
\end{figure}

\noindent{\it {\clr Determine the dispersive phase in arbitrary dimensions.}}--- Here we provide an analytical method to determine the dispersive-to-proliferative transition point based on a deforming contour method, which is applicable in arbitrary dimensions. 
We illustrate this method using the following $\mathcal{PT}$ symmetric non-Bloch Hamiltonian:
\begin{equation}
    \mathcal{H}_{\text{1D}}(z)= (1-\gamma)z + (1+\gamma)z^{-1} + \alpha(z^2+z^{-2}),
\end{equation}
where $\gamma,\alpha \geq 0$ is assumed and $z=e^{ik}$. 
Applying Eq.~(\ref{eq:keyequality}), we have
\begin{equation}
    \bar{G}_{\text{1D}}(t)=\frac{1}{2\pi i}\oint_{|z|=1}\frac{e^{-i\mathcal{H}_{\text{1D}}(z)t}}{z} \, \dd z.
\label{eq:loop_integral}
\end{equation}
Because $\left|\bar{G}_{\text{1D}}\right|$ either decays to zero or grows exponentially as $t\to +\infty$, we only need to probe whether $\bar{G}(t)$ is bounded to determine if the system is in the dispersive phase. 
For this purpose, we first identify the region on the complex $z$ plane where $\Im\mathcal{H}_{\text{1D}}(z)\leq 0$, and then ask whether the integral contour of Eq.~\eqref{eq:loop_integral}, i.e., the unit circle $|z|=1$, can be adiabatically deformed into this region without hitting any poles. 
This criterion is a sufficient condition for the dispersive phase, because if such a contour deformation can be achieved, $|e^{-i\mathcal{H}_{\text{1D}}(z)t}|$ is bounded by 1 for all $t\geq 0$ on the contour, which implies that $\bar{G}(t)$ is bounded at $t\rightarrow+\infty$. 
Using this method, we compute the dispersive-to-proliferative transition points of this model, as shown in Fig.~\ref{fig:4}(a). 

Our method is straightforwardly applicable to higher dimensions. 
Now consider a 2D $\mathcal{PT}$ symmetric non-Bloch Hamiltonian [as illustrated in Fig.~\ref{fig:4}(b)]
\begin{equation}\label{eq:PT2Dexample}
\begin{aligned}
    \mathcal{H}_{\text{2D}}(z_{x},z_{y})=&(1+\gamma)(z_{x}+z_{y})+(1-\gamma)(z_{x}^{-1}+z_{y}^{-1})\\
    &+\alpha(z_{x}+z_{x}^{-1})(z_{y}+z_{y}^{-1}),
\end{aligned}
\end{equation}
where $z_{i}:=e^{ik_{i}}$ for $i=x,y$, and $\gamma, \alpha \geq 0$. 
We have
\begin{equation}\label{eq:barGinPT2Dexample}
    \bar{G}_{\text{2D}}(t)=\frac{1}{(2\pi i)^{2}}\oint_{|z_x|=1}\oint_{|z_y|=1}\frac{e^{-i\mathcal{H}_{\text{2D}}(z_{x},z_{y})t}}{z_{x}z_{y}} \, \dd z_{x}\dd z_{y}.
\end{equation}
If for any $\theta\in[0,2\pi]$, we can adiabatically deform the integral contour $|z_{y}|=1$ into $\Im\mathcal{H}_{\text{2D}}(r(\theta)e^{i\theta},z_y)\leq 0$, where $r(\theta)>0$ is some smooth function of $\theta$, then the same deforming contour argument shows that $|\bar{G}_{\text{2D}}(t)|$ is bounded.
Based on this method, we find the model has a non-zero dispersive-to-proliferative transition point $\gamma_{c}$ for $\alpha<0.5$.
We further plot the phase diagram in Fig.~\ref{fig:4}(c).
The phase diagram for dispersive-to-proliferative transition is sharply contrasted with the universal zero threshold ($\gamma^{*}=0$) of real-to-complex transition in the higher-dimensional OBC spectrum~\cite{Song_2022}. 
$|\bar{G}_{\text{2D}}(t)|$ decays to zero when $\gamma<\gamma_{c}$ and grows exponentially at large $t$ when $\gamma>\gamma_{c}$ [Fig.~\ref{fig:4}(d)], which validates this method. 

\noindent{\it {\clr Conclusion.}}---In this paper, we formulate a universal spectral moments theorem, applicable to any systems with finite ranged couplings. We demonstrate that the spectral moments are invariant with respect to boundary conditions. Hence they form a new class of bulk quantities and strongly constrains the OBC spectrum. We further give an analytical expression for the spectral moments based on the spectrum under PBC which are easy to compute, and an analytical expression for the average eigenvalue of the time-evolution operator in the thermodynamic limit, $\bar{G}(t)$. The boundedness of $\bar{G}(t)$ suggests a new phase of dispersiveness for non-Hermitian band systems. We further proposed a deforming contour method on the boundedness of $\bar{G}(t)$, which serves as a criterion for dispersiveness in arbitrary dimension. Our work demonstrates the feasibility of constructing bulk quantities to study the physical properties of non-Hermitian systems. The existence of a non-trivial 2D $\mathcal{PT}$ symmetric Hamiltonian where small non-Hermiticity preserves dispersiveness opens up new avenues for the study of wavepacket dynamics in higher dimensions. 

\noindent{\it {\clr Acknowledgements.}}---The authors thank Y.M. Hu for valuable discussions. 
This work was supported in part by the Office of Naval Research MURI N00014-20-1-2479 (N.C., C.S., K.Z., X.M. and K.S.), the National Science Foundation through the Materials Research Science and Engineering Center at the University of Michigan, Award No. DMR-2309029 (N.C., X.M. and K.S.), the Gordon and Betty Moore Foundation Grant No.N031710 (K.S.), and the Office Navy Research Award N00014-21-1-2770 (K.S.).

\appendix
\setcounter{equation}{0}  
\renewcommand{\theequation}{A\arabic{equation}}
\begin{widetext}
\section{Appendix A. Convergence rate of spectral moments}
\begin{lemma}\label{lemma:1}
    Let $\lambda_{1},...,\lambda_{N_{r}}$ be the eigenvalues of $\hat{H}_{r}$. For all $n\in\mathbb{N}^{+}$ we have
    \begin{equation}\label{eq:convergencespeed}
        \abs{\frac{1}{N_{r}}\sum_{j=1}^{N_{r}}\lambda_{i}^{n}-\frac{1}{mV}\int_{\bs{k}\in BZ}\Tr (H(\bs{k})^{n})\,\dd\bs{k}}=\mathcal{O}(r),
    \end{equation}
    where $m$ is the number of degree of freedom per unit cell and $V$ is the volume of the Brillouin zone.
\end{lemma}
\begin{proof}
    We have
    \begin{equation}\label{eq:traceloopformula}
        \frac{1}{N_{r}}\sum_{j=1}^{N_{r}}\lambda_{i}^{n}=\frac{\Tr H_{r}^{n}}{N_{r}}=\frac{\sum_{L}\omega(L)}{N_{r}},
    \end{equation}
    where the last sum is over all loops of length $n$ in $r\Gamma\cap\Omega$ and $\omega(L)$ is the weight of the corresponding loop.
    Since
    \begin{equation}
        \sum_{L}\omega(L)=\sum_{s}\sum_{L_{s}}\omega(L_{s}),
    \end{equation}
    where the first sum is over all nodes in $r\Gamma\cap\Omega$ and the second sum is over all loops of length $n$ with starting point $s$.
    Then, the number of points in $r\Gamma\cap\Omega$ that has at least a loop of length $n$ touching the boundary (denote this set of points as $\partial\Omega(n)$) is proportional to $nB_{r}$, where $B_{r}\propto r^{1-d}$ is the number of nodes on the boundary $\partial\Omega$.
    So we have
    \begin{equation}\label{eq:splitsum}
        \frac{\sum_{s}\sum_{L_{s}}\omega(L_{s})}{N_{r}}=\frac{\sum_{s\in\partial\Omega(n)}\sum_{L_{s}}\omega(L_{s})}{N_{r}}+\frac{\sum_{s\notin\partial\Omega(n)}\sum_{L_{s}}\omega(L_{s})}{N_{r}}.
    \end{equation}
    Since $N_{r}\propto r^{-d}$, we have 
    \begin{equation}\label{eq:estimate1}
        \frac{\sum_{s\in\partial\Omega(n)}\sum_{L_{s}}\omega(L_{s})}{N_{r}}=\mathcal{O}(r).
    \end{equation}
    To estimate the second term in the right hand side of Eq.~\eqref{eq:splitsum}, let $\hat{H}_{R}$ be the Hamiltonian one $r\Gamma$ with periodic condition such that the number of nodes along on direction, $R$, is much larger than $n$. 
    Let $n_{R}$ be the number of unit cells contained in the graph defining $\hat{H}_{R}$.
    Then for any unit cell $U\in\Omega\setminus\partial\Omega(n)$, we have
    \begin{equation}\label{eq:pbctrace}
        \sum_{s\in U}\sum_{L_{s}}\omega(L_{s})=\frac{\Tr H_{R}^{n}}{n_{R}}.
    \end{equation}
    Let $u_{r}$ be the number of unit cells in the network $r\Gamma\cap\Omega\setminus\Omega(n)$, then we have
    \begin{equation}\label{eq:estimate2}
        \frac{N_{r}-mu_{r}}{N_{r}}=\mathcal{O}(r).
    \end{equation}
    Hence we see
    \begin{equation}\label{eq:estimate3}
        \frac{\sum_{s\notin\partial\Omega(n)}\sum_{L_{s}}\omega(L_{s})}{N_{r}}=\frac{u_{r}\sum_{s\in U}\sum_{L_{s}}\omega(L_{s})}{N_{r}}=\frac{mu_{r}-N_{r}}{N_{r}}\frac{\sum_{s\in U}\sum_{L_{s}}\omega(L_{s})}{m}+\frac{\sum_{s\in U}\sum_{L_{s}}\omega(L_{s})}{m}.
    \end{equation}
    From Eq.\eqref{eq:pbctrace} we have 
    \begin{equation}\label{eq:estimate4}
        \frac{\sum_{s\in U}\sum_{L_{s}}\omega(L_{s})}{m}=\frac{\Tr H_{R}^{n}}{mn_{R}}=\lim_{R\rightarrow\infty}\frac{\Tr H_{R}^{n}}{mn_{R}}=\frac{1}{mV}\int_{\bs{k}\in BZ}\Tr (H(\bs{k})^{n})\dd\bs{k}.
    \end{equation}
    Combining Eq.~\eqref{eq:estimate1}, Eq.~\eqref{eq:estimate2},Eq.~\eqref{eq:estimate3} and Eq.~\eqref{eq:estimate4} we obtain Eq.\eqref{eq:convergencespeed}.
\end{proof}

\section{Appendix B. Proof of Eq.~\eqref{eq:szegothm} in the case of $F(z)=e^{-izt}$}
\begin{proof}
    It is easy to see that $\exists \,c>0, \forall\, r>0, n>0, \bs{k}\in BZ$,  we have $|\lambda_{s}(\bs{k})|<c$ and $|\Tr H_{r}^{n}/N_{r}|<c^n$. Given $t\in\mathbb{R}$ and $\epsilon>0$, $\exists\, N>0$ such that $\forall\, q>N, \bs{k}\in BZ, r>0$, we have
    \begin{equation}
        \abs{e^{-i\lambda_{s}(\bs{k})t}-\sum_{j=1}^{q}\frac{(-i\lambda_{s}(\bs{k})t)^{j}}{j!}}<\epsilon,\; \abs{\frac{\Tr\; e^{-iH_{r}t}}{N_{r}}-\sum_{j=1}^{q}\frac{(-it)^{j}\Tr H_{r}^{j}}{j!N_{r}}}<\epsilon
    \end{equation}
    Fix a $q>N$, we have
    \begin{equation}
        \abs{\frac{\Tr \; e^{-iH_{r}t}}{N_{r}}-\frac{1}{mV}\sum_{s=1}^{m}\int_{\bs{k}\in BZ}e^{-i\lambda_{s}(\bs{k})t}\dd\vec{k}}\leq A(r)+B(r)+C(r),
    \end{equation}
    where $A(r), B(r), C(r)$ are defined as
    \begin{equation}
    \begin{aligned}
        &A(r)=\abs{\frac{\Tr \;e^{-iH_{r}t}} {N_{r}}-\sum_{j=1}^{q}\frac{(-it)^{j}\Tr H_{r}^{j}}{j!N_{r}}},\\
        &B(r)=\frac{1}{mV}\abs{\sum_{s=1}^{m}\int_{\bs{k}\in BZ}\left(e^{-i\lambda_{s}(\bs{k})t}-\sum_{j=1}^{q}\frac{(-i\lambda_{s}(\bs{k})t)^{j}}{j!}\right)\dd\bs{k}},\\
        &C(r)=\abs{\sum_{j=1}^{q}\frac{(-it)^{j}\Tr H_{r}^{j}}{j!N_{r}}-\sum_{j=1}^{q}\sum_{s=1}^{m}\frac{1}{mV}\int_{\bs{k}\in BZ}\frac{(-i\lambda_{s}(\bs{k})t)^{j}}{j!}\dd\bs{k}}.
    \end{aligned}
    \end{equation}
    Note that $0\leq A(r)<\epsilon$, $0\leq B(r)<\epsilon$ and $C(r)\rightarrow 0$ as $r\rightarrow 0$ (by lemma~\ref{lemma:1}). Taking a upper limit of $r\rightarrow 0$, we have
    \begin{equation}
        \limsup_{r\rightarrow 0}\abs{\frac{\Tr\;e^{-iH_{r}t}}{N_{r}}-\frac{1}{mV}\sum_{s=1}^{m}\int_{\bs{k}\in BZ}e^{-i\lambda_{s}(\bs{k})t}\dd\bs{k}}<2\epsilon
    \end{equation}
    Since $\epsilon$ is arbitrary, we arrive at Eq.~\eqref{eq:keyequality}.
\end{proof}

\section{Appendix C. Proof of Corollary 1}
\begin{proof} 
If $\bar{G}(t_{0})>1$ for some $t_{0}>0$, from Eq.~\eqref{eq:keyequality}, we know $\exists \,r_{0}>0$ and $\epsilon>0$, $\forall \,0<r<r_{0}, |\Tr\;e^{-iH_{r}t_{0}}/N_{r}|>e^{\epsilon t_{0}}$. 
If the absolute values of imaginary parts of all the eigenvalues of $H_{r}$ are all less than $\epsilon$, then $|e^{-i\lambda t_{0}}|<e^{\epsilon t_{0}}$ holds for all eigenvalue $\lambda$ of $H_{r}$. 
Using the triangle inequality for absolute values and the fact that the trace of a matrix equals the sum of all its eigenvalues, we have $|\Tr\;e^{-iH_{r}t_{0}}/N_{r}|<e^{\epsilon t_{0}}$, which is a contradiction.
The proof for the case $|\bar{G}(t_{0})|>1$ for some $t_{0}<0$ is similar. 
\end{proof} 
\end{widetext}

\bibliography{refs_main}
\bibliographystyle{apsrev4-2}
\end{document}